\documentclass[conference]{IEEEtran}
\IEEEoverridecommandlockouts
\def\BibTeX{{\rm B\kern-.05em{\sc i\kern-.025em b}\kern-.08em
    T\kern-.1667em\lower.7ex\hbox{E}\kern-.125emX}}

\usepackage{graphicx}
\usepackage{color}
\usepackage{cite}
\usepackage{amsmath,amsthm}
\usepackage{amssymb}
\usepackage{comment}
\usepackage{mathtools}
\usepackage{tabulary}
\usepackage{acronym}
\usepackage{upgreek}
\usepackage{algorithm}
\usepackage{algpseudocode}
\usepackage{placeins}
\usepackage[super]{nth}
\usepackage{dsfont}
\usepackage{hyperref}
\usepackage{subcaption}

\newtheorem{theorem}{\bf Theorem}

\newtheorem{remark}{Remark}

\begin{document}

%\title{Asymptotic analysis of $1$-bit RIS-empowered MIMO systems: Closed-form tuning and performance guarantees}
\title{MIMO Communications with $1$-bit RIS: Asymptotic Analysis and Over-the-Air Channel Diagonalization\thanks{This work has been supported by the SNS JU project TERRAMETA
under the EU’s Horizon Europe research and innovation program under Grant
Agreement No 101097101, including top-up funding by UKRI under the UK
government’s Horizon Europe funding guarantee.
}}

\author{\IEEEauthorblockN{Panagiotis Gavriilidis, Kyriakos Stylianopoulos, and George C. Alexandropoulos}
\IEEEauthorblockA{Department of Informatics and Telecommunications,
National and Kapodistrian University of Athens, Greece\\
emails: \{pangavr, kstylianop, alexandg\}@di.uoa.gr
}\vspace{-0.8cm}}

\maketitle

\begin{abstract}
This paper presents an asymptotic analysis of Multiple-Input Multiple-Output (MIMO) systems assisted by a \(1\)-bit Reconfigurable Intelligent Surface (RIS) under Ricean fading conditions. Using random matrix theory, we show that, in the asymptotic regime, the dominant singular values and vectors of the transmitter–RIS and RIS–receiver channels converge to their deterministic Line-of-Sight (LoS) components, almost irrespective of the Ricean factors. This enables RIS phase configuration using only LoS information through a closed-form Sign Alignment (SA) rule that maximizes the channel gain. Furthermore, when the RIS is asymptotically larger than the transceiver arrays, proper RIS configuration can render the end-to-end MIMO channel in the capacity formula asymptotically diagonal, thereby eliminating inter-stream interference and enabling Over-The-Air (OTA) spatial multiplexing without channel knowledge at the transmitter. Building on this result, a waterfilling-inspired SA algorithm that allocates RIS elements to spatial streams, based on the asymptotic singular values and statistical channel parameters, is proposed. Simulation results validate the theoretical analyses, demonstrating that the proposed schemes achieve performance comparable to conventional Riemannian manifold optimization, but with orders of magnitude lower runtime.
%Our findings will lay the groundwork for scalable, low-complexity RIS control with asymptotic performance guarantees.
\end{abstract}
%\vspace{-0.2cm}
\vspace{-0.3cm}
\section{Introduction}

Reconfigurable Intelligent Surfaces (RISs) have emerged as a transformative solution to dynamically shape wireless propagation environments with low power consumption and hardware complexity \cite{EURASIP_RIS_paper,10596064}. By manipulating predominantly the phase of incident signals, RISs can substantially improve coverage and spectral efficiency, especially in scenarios with blocked or attenuated direct paths \cite{AJG2024}. To counter double path loss, while exploiting their low cost scalability, RISs are envisioned to consist of thousands of unit cells, far exceeding traditional array sizes. However, two major challenges arise: \textit{i}) each element typically supports only $1$-bit phase control (i.e., $\{0, \pi\}$) \cite{mmWave_PIN_RIS_power,power_consumption_practical_modeling}, rendering the configuration problem an NP-hard one; and \textit{ii}) acquiring high-dimensional Channel State Information (CSI) becomes prohibitively complex~\cite{Tsinghua_RIS_Tutorial_ALL}. 

Recent works, such as \cite{stylianopoulos2024signalignment}, have tackled \(1\)-bit RIS optimization for signal power maximization in Single-Input-Single-Output (SISO) communication systems. However, multi-stream Multiple-Input Multiple-Output (MIMO) scenarios under such discrete-phase constraints remain largely unexplored. There exist, nonetheless, several asymptotic studies focusing on ergodic capacity analyses with continuous phase control, where correlation matrices are derived from the angle-spread properties of the wireless channel. In partcular, the authors in \cite{moustakas2021capacity} and \cite{MA_TWC_2024} proposed
ergodically optimal RIS designs for mult-RIS-aided MIMO and the MIMO Multiple-Access Channel
(MAC), respectively. Complementarily, several large-system analyses have been developed to obtain deterministic equivalents of the ergodic rate in RIS-assisted systems using random matrix theory and statistical physics tools. 
In \cite{Zhang2021LargeSystemRIS_statistical_CSIT}, the achievable rate of RIS-assisted MIMO systems under correlated Ricean fading was characterized in closed form, and an Alternate Optimization (AO) method was proposed for joint transmit covariance and RIS phase-shift design. 
To account for limited scattering, \cite{Zhang2022LargeSystemRIS_Channel_rankDef} examined channel rank deficiency, provided an ergodic rate approximation, and proposed an AO scheme for joint optimization. 
Using replica theory, \cite{Xu2021LargeSystemRIS_CorrRicianFading} analyzed RIS-assisted MIMO MAC under spatially correlated Ricean fading, derived deterministic equivalents for the asymptotic sum rate, and developed an AO-based joint transceiver and RIS design. 
Finally, \cite{You2021LargeSystemRIS_PartialCSI_uplink} considered RIS-assisted multiuser MIMO uplink transmission with partial CSI, formulating an energy-efficiency maximization problem solved via deterministic equivalents and block-coordinate optimization.

In this paper, different from the vast majority of the state of the art that considers continuous RIS phase configuration control and analyzes ergodic (average)
performance metrics, we focus on studying instantaneous performance guarantees in $1$-bit RIS-aided MIMO systems. In contrast to the AO- or gradient-based schemes developed in \cite{Xu2021LargeSystemRIS_CorrRicianFading,You2021LargeSystemRIS_PartialCSI_uplink,Zhang2022LargeSystemRIS_Channel_rankDef,Zhang2021LargeSystemRIS_statistical_CSIT}, we herein present analytical closed-form frameworks for both channel gain and capacity maximization under discrete RIS response constraints and without the requirement for full CSI. It is proven that the burden of
channel diagonalization can be shifted away from the transmitter
and effectively realized OTA through the RIS. These results yield new theoretical insights into the asymptotic
diagonalization properties and scaling laws of large $1$-bit RIS-assisted
MIMO systems.
\vspace{-0.2 cm}
\subsection{Contributions}
\vspace{-0.1 cm}
This work develops a theoretical framework for the asymptotic optimization of RIS-aided MIMO systems with \(1\)-bit phase control under Ricean fading conditions. The main contributions of this paper are summarized in the following.
\begin{itemize}
    \item We prove that the dominant singular values and vectors of the Transmitter (TX) to the RIS as well as of the RIS to the Receiver (RX) channels converge to their deterministic Line-of-Sight (LoS) components (almost) irrespective of the Ricean \(K\)-factors.
    \item A closed-form Sign Alignment (SA) RIS design is derived that can asymptotically maximize the end-to-end channel gain, requiring only the sign patterns of the LoS channel components. 
    \item It is proven that, when the size of the RIS panel grows asymptotically larger than the size of the transceivers, the effective RIS-parametrized MIMO channel becomes asymptotically diagonal, allowing interference-free spatial multiplexing without TX-side CSI or precoding.
    \item Capitalizing on the RIS-enabled OTA channel diagonalization, a low complexity Waterfilling-inspired SA (W-SA) scheme is presented that allocates unit elements of the RIS to spatial streams based on the channel's asymptotic singular values and statistical CSI.
\item It is shown that both proposed designs scale linearly and without iterations w.r.t. the number of RIS elements, and they achieve performance comparable to full-CSI iterative methods with orders of magnitude lower runtime.

\end{itemize}

\section{System Model}

We consider an RIS-aided MIMO communication system, where the TX and RX are equipped with \(N_{\rm T}\) and \(N_{\rm R}\) antennas, respectively, while the RIS consists of \(N_{\rm S}\) reflecting elements. Let \(\mathbf{x} \in \mathbb{C}^{N_{\rm T} \times 1}\) denote the transmitted signal, \(\mathbf{y} \in \mathbb{C}^{N_{\rm R} \times 1}\) the received signal, and \(\mathbf{n} \in \mathbb{C}^{N_{\rm R} \times 1} \sim \mathcal{CN}(\mathbf{0}_{N_{\rm R} \times 1}, \sigma^2 \mathbf{I}_{N_{\rm R}})\) the additive thermal noise at the receiver, which are related as:
\begin{equation}\label{eq: received signal}
    \mathbf{y} = \sqrt{{\rm P_L}} \mathbf{H}_{\rm R}^{\rm H} \boldsymbol{\Phi} \mathbf{H}_{\rm T} \mathbf{x} +  \mathbf{n}.
\end{equation}
In this expression, \(\mathbf{H}_{\rm T} \in \mathbb{C}^{N_{\rm S } \times N_{\rm T}}\) is the TX-RIS channel with \(E[{\rm Tr}(\mathbf{H}_{\rm T} \mathbf{H}_{\rm T}^{\rm H})] = N_{\rm T} N_{\rm S}\), and \(\mathbf{H}_{\rm R}\) denotes the RIS-RX channel with \(E[{\rm Tr}(\mathbf{H}_{\rm R} \mathbf{H}_{\rm R}^{\rm H})] = N_{\rm R} N_{\rm S}\). The diagonal matrix \(\boldsymbol{\Phi} \in \mathbb{C}^{N_{\rm S} \times N_{\rm S}}\) denotes the RIS's reflection matrix, with each non-zero entry being \(1\)-bit reconfigurable, while \({\rm P_L}\) is the path loss term of the cascaded links. The transmitted signal satisfies the power constraint \({\rm Tr}(\mathbf{x} \mathbf{x}^{\rm H}) \leq P_T\), where \(P_T\) is the available power at the TX. We make the common assumption that the direct TX-RX channel is blocked, which constitutes the generic scenarios where RIS's role can be prominent~\cite{GA_AoI_BoI}.

Let us express the cascaded channel \(\mathbf{\tilde{H}} \triangleq \mathbf{H}^{\rm H}_{\rm R} \boldsymbol{\Phi} \mathbf{H}_{\rm T}\) in a more tractable form that highlights the impact of the RIS phase profile on its singular values. To this end, we apply the Singular Value Decomposition (SVD) to \(\mathbf{H}_{\rm R}^{\rm H}\) and \(\mathbf{H}_{\rm T}\), yielding \(\mathbf{H}_{\rm R}^{\rm H} = \mathbf{U}_{\rm R} \mathbf{D}_{\rm R} \mathbf{V}_{\rm R}^{\rm H}\) and \(\mathbf{H}_{\rm T} = \mathbf{U}_{\rm T} \mathbf{D}_{\rm T} \mathbf{V}_{\rm T}^{\rm H}\), where \(\mathbf{U}_x\) and \(\mathbf{V}_x\) are their left and right singular matrices, and \(\mathbf{D}_x\) is a diagonal matrix with singular values in decreasing order, i.e., \([\mathbf{D}_{\rm x}]_{1,1}>[\mathbf{D}_{\rm x}]_{2,2}>\ldots> [\mathbf{D}_{\rm x}]_{N_x,N_x}\). The subscript ``\({\rm x}\)'' stands for either ``\({\rm T}\)'' or ``\({\rm R}\)'', indicating respectively TX and RX. Then, the RIS-parametrized end-to-end channel  can be rewritten as follows:
\begin{equation}
    \mathbf{\tilde{H}} = \mathbf{U}_{\rm R} \left(\mathbf{D}_{\rm R}  \mathbf{V}_{\rm R}^{\rm H} \boldsymbol{\Phi} \mathbf{U}_{\rm T}  \mathbf{D}_{\rm T} \right)\mathbf{V}_{\rm T}^{\rm H}.
\end{equation}
It can be easily shown that the central term admits the following entrywise structure:
\begin{equation}\label{eq:SVD form}
    \left[\mathbf{D}_{\rm R} \mathbf{V}_{\rm R}^{\rm H} \boldsymbol{\Phi} \mathbf{U}_{\rm T} \mathbf{D}_{\rm T}\right]_{i,j} 
    = d_{{\rm R},i} d_{{\rm T},j}\left( \mathbf{v}_{{\rm R},i}^{*} \odot \mathbf{u}_{{\rm T},j} \right)^{\rm T} \boldsymbol{\phi},
\end{equation}
where \(\mathbf{v}_{{\rm R},i}\) and \(\mathbf{u}_{{\rm T},j}\) are the \(i\)-th (\(i=1,\ldots,N_{\rm R}\)) and \(j\)-th (\(j=1,\ldots,N_{\rm T}\)) columns of \(\mathbf{V}_{\rm R}\) and \(\mathbf{U}_{\rm T}\), respectively, \(\boldsymbol{\phi} \) is the vector containing the diagonal entries of \(\boldsymbol{\Phi}\), and \(d_{{\rm R},i}\) and \(d_{{\rm T}, j}\) denote the \(i\)-th and \(j\)-th diagonal entries of \(\mathbf{D}_{\rm R}\) and \(\mathbf{D}_{\rm T}\). 
%In the sequel we will examine the Ricean and pure Line of Sight (LoS) channel conditions and optimize the RIS configuration accordingly providing asymptotic guarantees for the instantaneous performance at each case. 

\section{Channel Gain Maximization}\label{sec: Channel Gain Maximization}
We herein focus on the maximization of the \(1\)-bit RIS-parametrized channel gain in the cases of pure Line-of-Sight (LoS) and Ricean channel conditions, which is written as:
\begin{align}\label{eq: OP formulation original}
    \mathcal{OP}_1:\,&\max_{\boldsymbol{\phi}} || \mathbf{\tilde{H}}||^2_{\rm F} = || \mathbf{D}_{\rm R} \mathbf{V}_{\rm R}^{\rm H} \boldsymbol{\Phi} \mathbf{U}_{\rm T} \mathbf{D}_{\rm T}||^2_{\rm F} \nonumber\\
    &\text{s.t.} \,\,\, \boldsymbol{\phi}\in\{1,-1\}^{N_{\rm S}}.
\end{align}
The equality in the objective holds due to \(\mathbf{U}_{\rm R}\) and \(\mathbf{V}_{\rm T}\) being square unitary matrices, thus the latter do not affecting the Frobenius norm of the cascaded channel. By substituting \eqref{eq:SVD form}, \(\mathcal{OP}_1\)'s objective can be reformulated as follows:
\begin{equation}\label{eq: channel magnitude with phi}
    || \mathbf{\tilde{H}}||^2_{\rm F}=\sum_{i=1}^{N_{\rm R}} \sum_{j=1}^{N_{\rm T}} d^2_{{\rm R},i} d^2_{{\rm T},j} |\left( \mathbf{v}_{{\rm R},i}^{*} \odot \mathbf{u}_{{\rm T},j} \right)^{\rm T} \boldsymbol{\phi}|^2.
\end{equation}
By observing this formulation we can see that, if \(\boldsymbol{\phi}\) is configured to align the phases for one of the inner terms, it can individually maximize them. For example to maximize the \((i,j)\)-th term of the sum the optimal \(\boldsymbol{\phi}\) would be \(\boldsymbol{\phi}=\exp\left(\jmath\angle\left( \mathbf{v}_{{\rm R},i}^{*} \odot \mathbf{u}_{{\rm T},j} \right)^{*}\right)\).
%, leading to \(\left( \mathbf{v}_{{\rm R},i}^{*} \odot \mathbf{u}_{{\rm T},j} \right)^{\rm T} \boldsymbol{\phi}=1\), since \(\mathbf{v}_{{\rm R},i}\) and \(\mathbf{u}_{{\rm T},j}\) are orthonormal vectors with unitary Euclidean norms. 
However, this approach has the following issues: \textit{i}) the phase configuration \(\boldsymbol{\phi}\) is only one, while the number of terms that need to be optimized are \(N_{\rm R}N_{\rm T}\); and \textit{ii}) we consider the RIS to be restricted to \(1\)-bit configuration. To overcome the latter, we will employ the SA method \cite{stylianopoulos2024signalignment}, which is specifically derived for \(1\)-bit RISs and capitalizes on the fact that \(\boldsymbol{\phi}\) can optimally maximize either the amplitude of the real or imaginary part of \(\left( \mathbf{v}_{{\rm R},i}^{*} \odot \mathbf{u}_{{\rm T},j} \right)\), by aligning the signs of either its real or imaginary part. It can be easily seen that, regardless of the structure of \(\left( \mathbf{v}_{{\rm R},i}^{*} \odot \mathbf{u}_{{\rm T},j} \right)\), at least a quarter of the maximum, which is only attainable with a continuous phase shift, is guaranteed \cite[Corollary~1]{stylianopoulos2024signalignment}. These performance guarantees motivate us further towards using this algorithm; however, issue \(\textit{i}\)) persists. In the following we will study two channel conditions: pure LoS, and Ricean, drawing conclusions on the phase profile selection for each case and its asymptotic performance guarantees. 

\subsection{Pure LoS Conditions}

Under pure LoS conditions, \(\mathbf{H}_{\rm T}\) and \(\mathbf{H}_{\rm R}^{\rm H}\) can be written with respect to two steering vectors each; the one emanating (impinging) from (towards) the TX (RX) and the one impinging (reflected) on (from) the RIS, i.e., we can write:
\begin{equation}\label{eq: Pure LoS both sides}
    \mathbf{\tilde{H}} = \mathbf{a}_{\rm R}(\phi_{\rm R},\theta_{\rm R}) \mathbf{a}_{\rm S}^{\rm H}(\phi^{\rm O}_{\rm S},\theta^{\rm O}_{\rm S}) \mathbf{\Phi} \mathbf{a}_{\rm S}(\phi^{\rm I}_{\rm S},\theta^{\rm I}_{\rm S}) \mathbf{a}^{\rm H}_{\rm T}(\phi_{\rm T},\theta_{\rm T}),
\end{equation}
where \(\mathbf{a}_{\rm R}(\phi_{\rm R},\theta_{\rm R})\in \mathbb{C}^{N_{\rm R}\times 1},\, \mathbf{a}_{\rm T}(\phi_{\rm T}, \theta_{\rm T})\in \mathbb{C}^{N_{\rm T}\times 1}\), and \(\mathbf{a}_{\rm S}(\phi_{\rm S}, \theta_{\rm S})\in \mathbb{C}^{N_{\rm S}\times 1}\) denote the unnormalized steering vectors for UPAs, with the respective angles being the pair of elevation and azimuth angles of arrival from the TX \((\phi_{\rm T},\theta_{\rm T})\), departure to the RX \((\phi_{\rm R},\theta_{\rm R})\), and the impinging/outgoing angles corresponding to the RIS \((\phi^{\rm I}_{\rm S},\theta^{\rm I}_{\rm S})\)/\((\phi^{\rm O}_{\rm S},\theta^{\rm O}_{\rm S})\). 
%(which in general are not the same since RIS can extend beyond Snell's law).
It can be easily seen that both channels are rank-\(1\) matrices since they are created from the product of two vectors, meaning that all their singular values, except the first, are zero. To this end, \(\mathcal{OP}_1\)'s objective can be re-expressed with only the first summation term in \eqref{eq: channel magnitude with phi}.
% :
% \begin{equation}\label{eq: pure LoS channel maximization}
%     \max_{\boldsymbol{\phi}} ||\mathbf{\tilde{H}}||^2_{\rm F} = d^2_{{\rm R},1} d^{2}_{{\rm T},1} |\left( \mathbf{v}_{{\rm R},1}^{*} \odot \mathbf{u}_{{\rm T},1} \right)^{\rm T} \boldsymbol{\phi}|^2.
% \end{equation}
Thereafter, it can be concluded that, by phase aligning only the element-wise product of the principal singular vectors, the channel gain is maximized. Considering the structure of the channel in \eqref{eq: Pure LoS both sides}, this is equivalent to the following maximization:
\begin{equation}\label{eq: pure LoS channel maximization with steering vecs}
    \max_{\boldsymbol{\phi}} \left|\left(\mathbf{a}^{*}_{\rm S}(\phi^{\rm O}_{\rm S},\theta^{\rm O}_{\rm S})\odot \mathbf{a}_{\rm S}(\phi^{\rm I}_{\rm S},\theta^{\rm I}_{\rm S})\right)^{\rm T}\boldsymbol{\phi}\right|^2.
\end{equation}
Letting \(\mathbf{b}_{\rm S}\) denote the inner product of the two steering vectors: \(\mathbf{b}_{\rm S}\triangleq \mathbf{a}^{*}_{\rm S}(\phi^{\rm O}_{\rm S},\theta^{\rm O}_{\rm S})\odot \mathbf{a}_{\rm S}(\phi^{\rm I}_{\rm S},\theta^{\rm I}_{\rm S})\), then the optimal continuous configuration would be \(\boldsymbol{\phi} = \exp\left(\jmath\angle\mathbf{b}_{\rm S}^{*} \right)\), while, for the \(1\)-bit case, we employ the SA method, which yields a gain of at least \(0.25 N^2_{\rm S}\) (with \(N^{2}_{\rm S}\) being the optimal for the continuous case). Consequently,
%following \cite[Theorem~3]{stylianopoulos2024signalignment},
% \((\mathbf{b}_{\rm S}^{\rm T}\boldsymbol{\phi})^2\geq 0.25 N^{2}_{\rm S}\) which yields
it holds that \(||\mathbf{\tilde{H}}||^2_{\rm F}\geq 0.25 N_{\rm R} N_{\rm T}N^2_{\rm S}\), following the general rule of \(N^2_{\rm S}\) scaling of the channel gain.

\subsection{Principal Singular Value Hardening in Ricean Conditions}

\begin{figure*}[!t]
    \centering
    % First subfigure
    \begin{subfigure}[b]{0.32\textwidth}
        \centering
        \includegraphics[width=\textwidth]{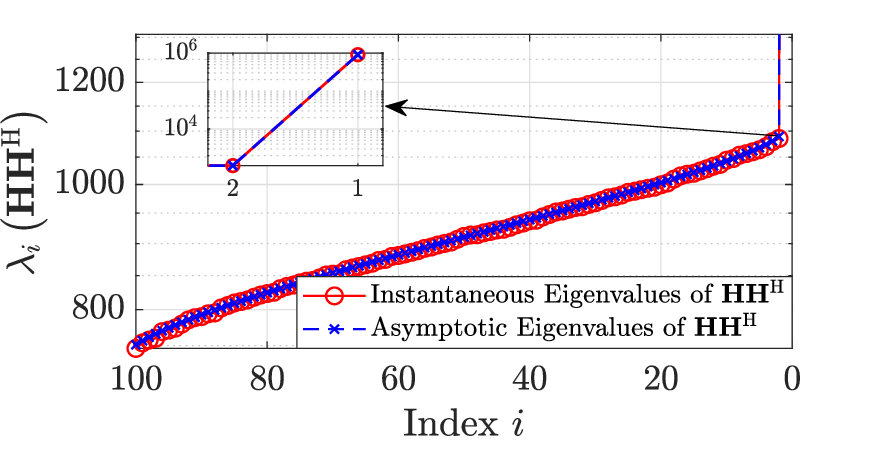} % Replace with your image
        \caption{}
        %\caption{Eigenvalues for a single instance of \(\mathbf{HH}^{\rm H}\).}
        \label{fig: eigenvalues of single instance}
    \end{subfigure}
    \hfill
    % Second subfigure
    \begin{subfigure}[b]{0.32\textwidth}
        \centering
        \includegraphics[width=\textwidth]{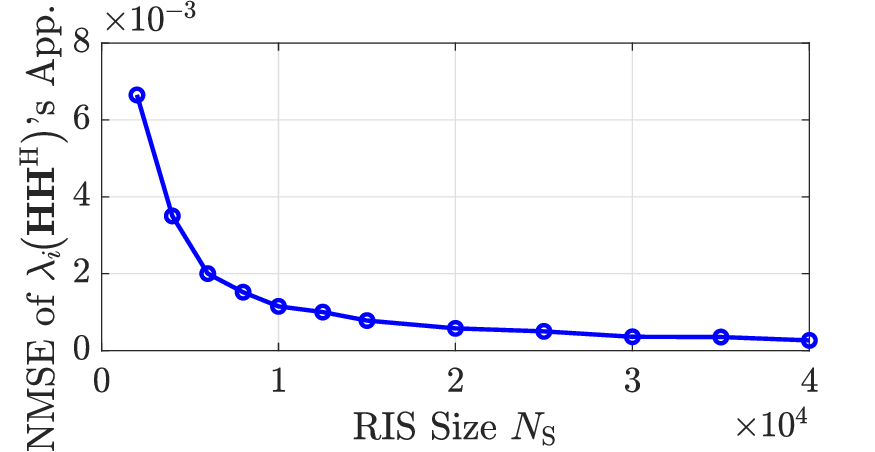} % Replace with your image
        \caption{}
        %\caption{MNSE of eigenvalues for increasing \(N\).}
        \label{fig: eigenvalues versus N}
    \end{subfigure}
    % Third subfigure
    \begin{subfigure}[b]{0.32\textwidth}
        \centering
        \includegraphics[width=\textwidth]{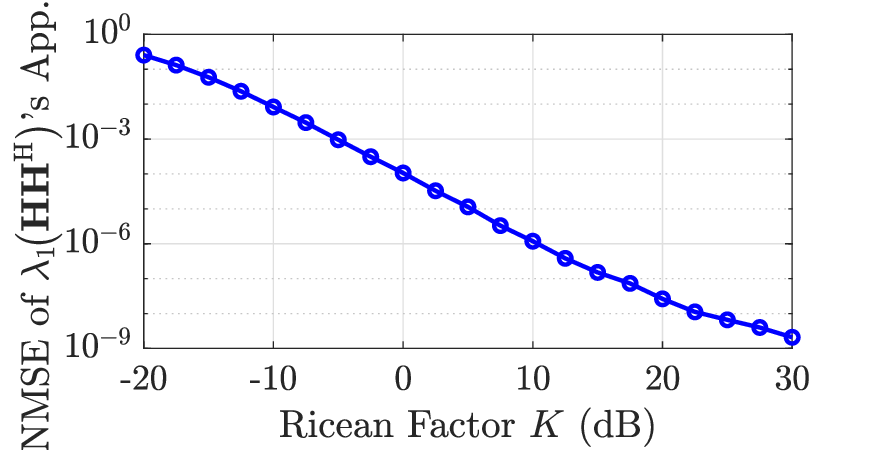} % Replace with your image
        \caption{}
        %\caption{MNSE of principal eigenvalue versus Ricean \(K\).}
        \label{fig: eigenvalues versus K}
    \end{subfigure}
    \caption{\small Convergence properties of a \(K\)-factor Ricean fading channel \(\mathbf{H} \in \mathbb{C}^{N_{\rm S} \times N_{\rm T}}\). Unless otherwise stated, the parameters \(N_{\rm S} = 10^4\), \(N_{\rm T} = 10^2\), and \(K = 10\,\mathrm{dB}\) were used. Figure~\ref{fig: eigenvalues of single instance} shows the eigenvalues of \(\mathbf{HH}^{\rm H}\) for a single channel instance which are compared with their respective asymptotics computed via \eqref{eq: convergence of all eigenvalues}; the Normalized Mean Squared Error (NMSE) of this asymptotic approximation is illustrated in Fig.~\ref{fig: eigenvalues versus N} for different sizes of RISs. In particular, the NMSE is computed for each eigenvalue and the aggregate error is depicted, revealing that asymptotic hardening takes place for increasing \(N_{\rm S}\). Figure~\ref{fig: eigenvalues versus K} depicts the NMSE between the principal eigenvalue of \(\mathbf{HH}^{\rm H}\) and \(K/(K+1)N_{\rm S}N_{\rm T}\), with the latter being the principal eigenvalue in the asymptotic regime according to Theorem~\ref{theorem: eigenvalue hardening}. The NMSE is plotted versus \(K\) starting with the value \(-20\)~dB where \(1/N_{\rm T}\) equals \(K\); recall that, at this point, convergence is lost, since it should hold that \(K/N_{\rm T}\to 0\). It can, however, be seen that, if \(K\geq 10N_{\rm T}^{-1}\), the NMSE drops below \(10^{-2}\) and converges to \(0\) for increasing \(K\).
%It is noted that for \(K < 1/S\) the principal eigenvector will follow the principal eigenvector of the random part of \(\mathbf{H}\), whch is why we do not show the NMSE before that value.
}
    \label{fig: simulations for eigenvalue theorems}
\end{figure*}

We consider the case where both \(\mathbf{H}^{\rm H}_{\rm R}\) and \(\mathbf{H}_{\rm T}\) follow the Ricean fading distribution with respective factors \(K_{\rm R}\) and \(K_{\rm T}\). To analytically assess the channel gain performance, we focus on the asymptotic regime where \(N_{\rm T}, N_{\rm R} \to \infty\). Moreover, since RISs are envisioned to operate at much larger scales than conventional transceivers, we examine the practically relevant case where \(\frac{N_{\rm R}}{N_{\rm S}}, \frac{N_{\rm T}}{N_{\rm S}} \to 0\), i.e., the RIS is asymptotically larger than the TX/RX. As \(\mathbf{H}_{\rm T}\) and \(\mathbf{H}^{\rm H}_{\rm R}\) share the same distribution, we study the former for notational simplicity, noting that all results apply identically to \(\mathbf{H}^{\rm H}_{\rm R}\). We begin by analyzing the asymptotic convergence of \(\mathbf{H}_{\rm T}\)'s principal singular values and vectors as the RIS phase configuration that maximizes channel gain, depends primarily on those.

\begin{theorem}\label{theorem: eigenvalue hardening}
Consider \(\mathbf{H}_{\rm T}\) to be Ricean distributed as: \(\mathbf{H}_{\rm T} = \frac{\sqrt{K_{\rm T}}}{\sqrt{K_{\rm T}+1}}\mathbf{A} + \frac{1}{\sqrt{K_{\rm T}+1}}\mathbf{B}\). If \(N_{\rm T},N_{\rm S}\to \infty\) and \(\frac{N_{\rm T}}{N_{\rm S}}\to 0\), then \(\sigma_1(\mathbf{H}_{\rm T})\to\sigma_1\left(\frac{\sqrt{K_{\rm T}}}{\sqrt{K_{\rm T}+1}}\mathbf{A}\right)\) and \(\mathbf{u}_{{\rm T},1}\to 1/\sqrt{N_{\rm S}} \,\mathbf{a}_{\rm S}(\phi^{\rm I}_{\rm S},\theta^{\rm I}_{\rm S})  \).
\end{theorem}

\begin{proof}
In \(\mathbf{H}_{\rm T}\)'s formulation, \(\mathbf{A} = \mathbf{a}_{\rm S}(\phi^{\rm I}_{\rm S},\theta^{\rm I}_{\rm S}) \mathbf{a}^{\rm H}_{\rm T}(\phi_{\rm T},\theta_{\rm T}) \) is a deterministic rank-\(1\) matrix modeling the LoS connection, while \(\mathbf{B}\) follows the Rayleigh distribution. To ease the notation, we define the auxiliary variables \(\mathbf{A}_{\rm K} \triangleq \frac{\sqrt{K}}{\sqrt{K+1}}\mathbf{A}\), \(\mathbf{B}_{\rm K} \triangleq \frac{1}{\sqrt{K+1}} \mathbf{B}\), and \(\mathbf{a}_{\rm S,T} \triangleq \mathbf{a}_{\rm S}(\phi^{\rm I}_{\rm S},\theta^{\rm I}_{\rm S})\). Let \(\sigma_i(\cdot)\) denote the \(i\)-th singular value of a matrix, then \(\sigma_1\left(\mathbf{A}_{\rm K}\right)=\frac{\sqrt{K_{\rm T}}}{\sqrt{K_{\rm T}+1}}\sqrt{N_{\rm S}N_{\rm T}}\). On the other hand, \(\mathbf{B}\mathbf{B}^{\rm H}\) is a Wishart matrix, and, since \(\frac{N_{\rm T}}{N_{\rm S}}\to 0\), according to \cite[Corollary 2.2.(c)]{DETTE2002largesteigenvalue} the largest eigenvalue of \(\mathbf{BB}^{\rm H}\) converges almost surely to \(\lambda_1(\mathbf{BB}^{\rm H})= 2 \sqrt{N_{\rm S}N_{\rm T}} + N_{\rm S}\), yielding \(\sigma_1(\mathbf{B})=\sqrt{2 \sqrt{N_{\rm S}N_{\rm T}} + N_{\rm S}}\). By applying Weyl's inequality, the following asymptotic result is deduced:
\begin{align}
\left|\sigma_1\left(\mathbf{H}_{\rm T}\right) -\sigma_{ 1}\left(\mathbf{A}_{\rm K} \right)\right|&\leq \sigma_1\left(\mathbf{B}_{\rm K}\right) \,\,\overset{\text{(a)}}{\Rightarrow} \\
\frac{\left|\sigma_1\left(\mathbf{H}_{\rm T}\right) - \sigma_{ 1}\left(\mathbf{A}_{\rm K}\right)\right|}{\sigma_1\left(\mathbf{A}_{\rm K}\right)} &\leq  0 \Rightarrow\,\, 
 \sigma_1(\mathbf{H}_{\rm T}) \to  \sigma_1(\mathbf{A}_{\rm K}).\nonumber
\end{align}
Herein, \(\text{(a)}\) holds at the asymptotic regime, since \(\lim_{N_{\rm S},N_{\rm T}\to \infty}\frac{\sigma_1\left(\mathbf{B}_{\rm K}\right)}{ \sigma_1\left(\mathbf{A}_{\rm K}\right)}= \lim_{N_{\rm S},N_{\rm T}\to \infty}  \frac{\sqrt{K}\sqrt{2 \sqrt{N_{\rm S}N_{\rm T}} + N_{\rm S}}}{\sqrt{N_{\rm S}N_{\rm T}}}=0\), which holds true under the mild assumption that \(\frac{K_{\rm T}}{N_{\rm T}}\to 0\). 

To prove that the principal (left-side) singular vector converges to the LoS steering vector, it suffices to show that the Rayleigh quotient of \(\mathbf{H}_{\rm T}\) with \(\mathbf{a}_{\rm S,T}\) equals \(\sigma_1^2(\mathbf{H}_{\rm T})\):
\begin{align}\label{eq: principal singular vector inequality}
        &\mathbf{a}_{\rm S,T}\mathbf{H}_{\rm T}\mathbf{H}^{\rm H}_{\rm T}\mathbf{a}^{\rm H}_{\rm S,T} =\nonumber\\
        &  \sigma_1^2\left(\mathbf{A}_{\rm K}\right) + \mathbf{a}^{\rm H}_{\rm S,T}\left( 2 {\rm Re}\left\{\mathbf{A}_{\rm K}\mathbf{B}_{\rm K}^{\rm H}\right\}+  \mathbf{B}_{\rm K}\mathbf{B}^{\rm H}_{\rm K}\right)\mathbf{a}_{\rm S,T}\leq \nonumber\\
        & \sigma_1^2\left(\mathbf{A}_{\rm K}\right) + \sigma_1^2\left(\mathbf{B}_{\rm K }\right) +  2 \sigma_1\left(\mathbf{B}_{\rm K}\right)\sigma_1(\mathbf{A}_{\rm K}).
    \end{align}
    In the asymptotic regime, \eqref{eq: principal singular vector inequality} converges to \(\sigma_1^2\left(\mathbf{A}_{\rm K}\right)\), since \(\lim_{N_{\rm T},N_{\rm S}\to \infty} \frac{\mathbf{a}_{\rm S,T}\mathbf{H}_{\rm T}\mathbf{H}^{\rm H}_{\rm T}\mathbf{a}^{\rm H}_{\rm S,T}}{\sigma_1^2(\mathbf{A}_{\rm K})} = 1 + \lim_{N_{\rm T},N_{\rm S}\to \infty} \frac{\mathbf{a}^{\rm H}_{\rm S,T}\left( 2 {\rm Re}\left\{\mathbf{A}_{\rm K}\mathbf{B}_{\rm K}^{\rm H}\right\}+  \mathbf{B}_{\rm K}\mathbf{B}^{\rm H}_{\rm K}\right)\mathbf{a}_{\rm S,T}}{\sigma_1^2(\mathbf{A}_{\rm K})}\), and using \eqref{eq: principal singular vector inequality}'s inequality, we conclude that \(\lim_{N_{\rm T},N_{\rm S}\to \infty} \frac{\mathbf{a}_{\rm S,T}\mathbf{H}_{\rm T}\mathbf{H}^{\rm H}_{\rm T}\mathbf{a}^{\rm H}_{\rm S,T}}{\sigma_1^2(\mathbf{A}_{\rm K})} = 1\), since \(\lim_{N_{\rm T},N_{\rm S}\to \infty} \frac{\mathbf{a}^{\rm H}_{\rm S,T}\left( 2 {\rm Re}\left\{\mathbf{A}_{\rm K}\mathbf{B}_{\rm K}^{\rm H}\right\}+  \mathbf{B}_{\rm K}\mathbf{B}^{\rm H}_{\rm K}\right)\mathbf{a}_{\rm S,T}}{\sigma_1^2(\mathbf{A}_{\rm K})}\leq \lim_{N_{\rm T},N_{\rm S}\to \infty} \frac{\sigma_1^2(\mathbf{B}_{\rm K}) + 2 \sigma_1(\mathbf{A}_{\rm K})\sigma_1(\mathbf{B}_{\rm K})}{\sigma_1^2(\mathbf{A}_{\rm K})} = 0\) when \(\frac{K_{\rm T}}{N_{\rm T}}\to 0\).
    
\end{proof}
%It follows directly that Theorem~\ref{theorem: eigenvalue hardening} also applies to \(\mathbf{H}^{\rm H}_{\rm R}\). 
Theorem~\ref{theorem: eigenvalue hardening} establishes an asymptotic hardening effect in Ricean fading channels, according to which the principal singular value and vector converge to those of the deterministic LoS component and its steering vector. 
%The same conclusion would hold when \(\frac{N_{\rm T}}{N_{\rm S}}\to \beta>0\), using the Marčenko–Pastur theorem \cite[Corollary~2.2.(a)]{DETTE2002largesteigenvalue}. 
This convergence is largely insensitive to the Ricean factor, provided that \(K_{\rm x}\gg 1/N_{\rm x}\), which is expected in the asymptotic regime for all but purely Rayleigh channels (i.e., when \(K_{\rm x}=0\)).

\subsection{Convergence of SVD Components in Ricean Conditions}
In this subsection, we aim to characterize the asymptotic behavior of the SVD components of \(\mathbf{H}_{\rm T}\) and \(\mathbf{H}^{\rm H}_{\rm R}\). 
Our goal is twofold: \textit{i)} to determine how the singular values beyond the dominant one converge; and \textit{ii)} to describe the distribution of the corresponding singular vectors. 
This analysis is essential for the RIS configuration optimization problem \(\mathcal{OP}_1\), since the achievable channel gain depends directly on these singular values and vectors. 
By understanding their limiting behavior, we can identify the asymptotically optimal RIS configuration strategy, and derive instantaneous performance guarantees.

We have previously demonstrated that the principal eigenvector of \(\mathbf{H}_{\rm T}\mathbf{H}^{\rm H}_{\rm T}\) (similarly for \(\mathbf{H}^{\rm H}_{\rm R}\mathbf{H}_{\rm R}\)) converges asymptotically to the eigenvector associated with the deterministic LoS component \(\mathbf{A}\mathbf{A}^{\rm H}\). Since \(\mathbf{A}\mathbf{A}^{\rm H}\) is a rank-\(1\) matrix, it contributes a single nonzero eigenvalue and its corresponding eigenvector, while all other orthogonal directions lie in its null space. Thus, asymptotically, the influence of \(\mathbf{A}\) on the rest of the eigenspectrum of \(\mathbf{H}_{\rm T}\mathbf{H}^{\rm H}_{\rm T}\) becomes negligible. Intuitively, this is because the dominant eigenvalue generated by \(\mathbf{A}\mathbf{A}^{\rm H}\) separates from the spectral bulk of the random component \(\mathbf{B}\mathbf{B}^{\rm H}\), leading to an effective decoupling of the corresponding eigenspaces. This behavior is consistent with results in random matrix theory, where it is known that the addition of a low-rank deterministic matrix to a Wishart-type random matrix preserves the empirical eigenvalue distribution of the random part, except for a finite number of outliers corresponding to the perturbation's rank~\cite{tao2012topics}. For our setting, however, we have the sum of non-Hermitian matrices (\(e.g., \mathbf{H}_{\rm T} = \mathbf{A}_{\rm K} + \mathbf{B}_{\rm K}\)), so we focus on the SVD, and similar separation results apply. In particular,~\cite[Theorem~2.8]{benaych2011singular} shows that the singular values and vectors of a low-rank deterministic perturbation detach from the bulk, which remains asymptotically governed by the random matrix \(\mathbf{B}\).
This becomes apparent from the inequality in~\cite[p.~11]{benaych2011singular}, according to which the singular values of \(\mathbf{H}_{\rm T}\) beyond the principal one are bounded between successive singular values of \(\mathbf{B}\), i.e., \(\sigma_{i+1}\left(\mathbf{B}\right) \leq \sigma_i\left(\mathbf{H}_{\rm T}\right)\leq \sigma_{i-1}\left(\mathbf{B}\right)\).

Based on the above, we model the rest of the SVD components of \(\mathbf{H}_{\rm T}\) (and \(\mathbf{H}^{\rm H}_{\rm R}\)) as identical in distribution to those of \(\mathbf{B}_{\rm K}\). Specifically, we are interested on the distributions of \(\mathbf{D}_{\rm T}\) and \(\mathbf{U}_{\rm T}\), which can be investigated by studying the eigenspectrum of \(\mathbf{B}_{\rm K}\mathbf{B}^{\rm H}_{\rm K} = \mathbf{U}_{\rm B}\mathbf{D}_{\rm B}^2\mathbf{U}^{\rm H}_{\rm B}\). Regarding the eigenvalue distribution of \(\mathbf{B}_{\rm K}\mathbf{B}^{\rm H}_{\rm K}\), we know that, for the case where \(N_{\rm S},N_{\rm T}\to \infty\) with \(N_{\rm T}/N_{\rm S}\to 0\) \cite[Corollary~2.2. (c)]{DETTE2002largesteigenvalue}, each eigenvalue converges to a non-random limit, which can be found via the roots of the Laguerre polynomial \(L_{N_{\rm S}}^{N_{\rm T}-N_{\rm S}}\left(2 \sqrt{N_{\rm S}N_{\rm T}} x + N_{\rm S} \right)\) \footnote{While \cite{DETTE2002largesteigenvalue} studied real Wishart matrices, it is well established that the limiting empirical eigenvalue distribution of both real and complex Wishart matrices converges to the same Marčenko–Pastur law.}. Exploiting this fact and Theorem~\ref{theorem: eigenvalue hardening}, one can compute the convergence of the singular values of \(\mathbf{H}_{\rm T}\) and \(\mathbf{H}^{\rm H}_{\rm R}\) as follows (assuming \(K_{\rm x}\gg 1/N_{\rm x}\)):
\begin{equation}
    d^2_{{\rm x},i}\to \begin{cases}
        \frac{K_{\rm x}}{K_{\rm x} + 1}N_{\rm S}N_{\rm x} & i=1 \label{eq: convergence of all eigenvalues}\\
       \frac{N_{\rm S} + 2 r\left[N_{\rm S}-i + 1\right]\sqrt{N_{\rm S}N_{\rm x}}}{K_{\rm x}+1} & i=2,\ldots,N_{\rm x}
    \end{cases},
\end{equation}
where \(r[1]<r[2]<\ldots<r[N_{\rm S}]\) denote the ordered roots of \(L_{N_{\rm S}}^{N_{\rm x}-N_{\rm S}}\left(2 \sqrt{N_{\rm S}N_{\rm x}} x + N_{\rm S} \right)\). Since there are \(N_{\rm x}\) eigenvalues, we are concerned only with the roots beyond the one at \(\left[N_{\rm S}-N_{\rm x}\right]\); these roots lie in \([-1,1]\). It is noted that, for \(K\ll 1/N_{\rm x}\), the singular values of the Ricean channel converge to those of its random counterpart. Theorem \ref{theorem: eigenvalue hardening} and the convergence of \eqref{eq: convergence of all eigenvalues} are corroborated in Fig.~\ref{fig: simulations for eigenvalue theorems}.

As for the eigenvectors of \(\mathbf{B}_{\rm K}\mathbf{B}^{\rm H }_{\rm K}\), since it is a Wishart matrix, its eigenspace lies uniformly in the set of unitary matrices (Haar measure \cite{tao2012topics}). Furthermore, it is known that, at each channel instance, \(\mathbf{B}_{\rm K}\mathbf{B}^{\rm H}_{\rm K}\) is of rank \(N_{\rm T}\), thus, from the set of \(N_{\rm S}\times N_{\rm S}\)-sized unitary matrices we are only concerned with a subpart of \(N_{\rm S}\times N_{\rm T}\) size. Let us define a matrix \(\mathbf{\Gamma}\in \mathbb{C}^{N_{\rm S}\times N_{\rm S}}\) which follows the Haar measure. Then, following \cite[Theorem~3]{jiang2006orthogonalmatrix}, its \(N_{\rm S}\times \frac{N_{\rm S}}{\log(N_{\rm S})}\) sized submatrix can be approximated by independent and identically distributed (i.i.d.) Gaussian normals. Furthermore, since we focus in the regime \(\frac{N_{\rm T}}{N_{\rm S}}\to 0\), \(N_{\rm T}<N_{\rm S}/\log(N_{\rm S})\) holds. Hence, we can assume that the elements of $\mathbf{U}_{\rm T}$, which corresponds to the $N_{\rm S} \times N_{\rm T}$ submatrix of a Haar-distributed matrix $\mathbf{\Gamma}$, can be modeled as i.i.d. Gaussian normals. Similarly, the same holds for $\mathbf{V}^{\rm H}_{\rm R}$; therefore, we can write $\forall i,j \geq 2$:
\begin{equation}\label{eq: Eiganvector Distribution}
    \mathbf{u}_{{\rm T},j}\,\,\text{and}\,\, \mathbf{v}^*_{{\rm R},i}\,\sim\mathcal{CN}(\mathbf{0}_{N_{\rm S}\times 1}, {1}/{\sqrt{N_{\rm S}}}\mathbf{I}_{N_{\rm S}}). 
\end{equation}
Based on these, we show in the following theorem that, when the RIS is configured to maximize one of the entries of  \(\mathbf{V}^{\rm H}_{\rm R}\mathbf{\Phi} \mathbf{U}_{\rm T}\), the remaining terms are asymptotically nulled. 

\begin{theorem}\label{theorem: channel magnitude convergence}
    Consider \(\mathbf{H}^{\rm H}_{\rm R}\) and \(\mathbf{H}_{\rm T}\) to be Ricean distributed and the 1-bit RIS to be optimized via SA to maximize \(|\left( \mathbf{v}_{{\rm R},i^{\star}}^{*} \odot \mathbf{u}_{{\rm T},j^{\star}} \right)^{\rm T} \boldsymbol{\phi}|\). If \(N_{\rm R}, N_{\rm T},N_{\rm S}\to \infty\) and \(\frac{N_{\rm T}}{N_{\rm S}},\frac{N_{\rm R}}{N_{\rm S}}\to 0\), then, \({|\left( \mathbf{v}_{{\rm R},i}^{*} \odot \mathbf{u}_{{\rm T},j} \right)^{\rm T} \boldsymbol{\phi}|^2} \to 0\,\,\forall i,j\,-\{i^{\star},j^{\star}\}\).
    %\(||\mathbf{\tilde{H}}||^{2}_{\rm F}\to d^2_{{\rm R},i^{\star}} d^2_{{\rm T},j^{\star}} |\left( \mathbf{v}_{{\rm R},i^{\star}}^{*} \odot \mathbf{u}_{{\rm T},j^{\star}} \right)^{\rm T} \boldsymbol{\phi}|^2\).
\end{theorem}
\begin{proof}
    Based on the distributions of the vectors in the element-wise product, the entries of $\mathbf{V}^{\rm H}_{\rm R}\mathbf{\Phi} \mathbf{U}_{\rm T}$ fall into the following three categories: \textit{i)} both are deterministic LoS components; \textit{ii)} both have independent complex Gaussian entries; and \textit{iii)} one is deterministic and the other Gaussian. Since we have established the independence among the singular vectors, we can assume that when the RIS is configured to optimize the $(i^{\star},j^{\star})$-th term, this configuration can be regarded as independent of the remaining terms. Moreover, both the LoS components and the Gaussian entries can be thought of as having a uniform phase over \([0,2\pi]\), thus, we can confidently assume that each RIS element has an equiprobable chance of having a \(-1\) or \(1\) response. For case \textit{i)}, it holds \(\mathbf{v}^{*}_{{\rm R},1}\odot \mathbf{u}_{{\rm T,1}}=\frac{1}{\sqrt{N_{\rm S}}} \mathbf{a}^{*}_{\rm S}\left(\phi^{\rm O}_{\rm S},\theta^{\rm O}_{\rm S}\right)\odot \frac{1}{\sqrt{N_{\rm S}}}\mathbf{a}_{\rm S}\left(\phi^{\rm I}_{\rm S},\theta^{\rm I}_{\rm S}\right)\). According to Lyapunov's Central Limit Theorem (CLT) \cite{billingsley1995probability}, we have for the weighted sum of i.i.d. random variables:
    \begin{align}
        &\left( \mathbf{v}_{{\rm R},i}^{*} \odot \mathbf{u}_{{\rm T},j} \right)^{\rm T} \boldsymbol{\phi} = \frac{1}{N_{\rm S}}\sum_{n=1}^{N_{\rm S}} [\mathbf{b}_{\rm S}]_n[\boldsymbol{\phi}]_n  = \nonumber\\
        & \frac{1}{\sum_{n=1}^{N_{\rm S}} |[\mathbf{b}_{\rm S}]_n|^2}\sum_{n=1}^{N_{\rm S}} [\mathbf{b}_{\rm S}]_n[\boldsymbol{\phi}]_n- E\left[ [\mathbf{b}_{\rm S}]_n[\boldsymbol{\phi}]_n\right] \to 0,
    \end{align}
    where the convergence is done with a rate of \(O\left(1/\sqrt{N_{\rm S}}\right)\).
    Similarly, we can prove that the other cases also converge to zero using the typical CLT theorem since the entries of the sum will be i.i.d. with zero mean. On the other hand, for the \((i^{\star},j^{\star})\)-th entry that is optimized via SA, it can be shown that, according to \cite[Theorems 2 and 3]{stylianopoulos2024signalignment}, \(|\left( \mathbf{v}_{{\rm R},i^{\star}}^{*} \odot \mathbf{u}_{{\rm T},j^{\star}} \right)^{\rm T} \boldsymbol{\phi}|\geq 0.5\) holds. Hence, the optimized term is asymptotically larger than the unoptimized ones, which concludes the proof.
\end{proof}

Following Theorem~\ref{theorem: channel magnitude convergence}, it can be seen that, optimizing one of the terms in \eqref{eq: channel magnitude with phi}, leads to the nulling of the rest in the asymptotic regime. Therefore, it can be deduced that the best choice would be to use all \(N_{\rm S}\) elements to perform SA on the first term: \(d^2_{{\rm R},1} d^2_{{\rm T},1} |\left( \mathbf{v}_{{\rm R},1}^{*} \odot \mathbf{u}_{{\rm T},1} \right)^{\rm T} \boldsymbol{\phi}|^2\), since, according to \eqref{eq: convergence of all eigenvalues}, \(d^2_{{\rm R},1} d^2_{{\rm T},1}\gg d^2_{{\rm R},i} d^2_{{\rm T},j}\) \(\forall i,j\,-\{1,1\}\). This choice yields the asymptotic channel gain \(||\mathbf{\tilde{H}}||^2_{\rm F}\to \alpha \frac{K_{\rm T}K_{\rm R}}{(1+K_{\rm T})(1+ K_{\rm R})} N^2_{\rm S}N_{\rm T}N_{\rm R} + \mathcal{O}(N_{\rm T}N_{\rm R})\) with \(\alpha \geq 0.25\) according to Theorems~\ref{theorem: eigenvalue hardening} and \ref{theorem: channel magnitude convergence}. Thus, with only the knowledge of the LoS components, the RIS elements can be configured to provide asymptotically optimal results in terms of channel gain, even in scenarios with low Ricean values, as long as it holds that \(K_{\rm R}\gg 1/N_{\rm R}\) and \(K_{\rm T}\gg 1/N_{\rm T}\).

\section{Capacity Optimization} \label{sec: Capacity Optimization}
Assuming Gaussian signaling with identity precoding (no CSI at the TX), that is, $\mathbf{x} \sim \mathcal{CN}(\mathbf{0}_{N_{\rm T}\times 1}, \frac{P_{\rm T}}{N_{\rm T}}\mathbf{I}_{N_{\rm T}})$, the capacity of the RIS-parametrized MIMO channel is computed as follows:
\begin{equation}\label{eq: capacity formula}
    \mathcal{C}(\boldsymbol{\phi})=\log_2 {\rm det}\left(\mathbf{I} + \frac{{\rm SNR}}{N_{\rm T}} \mathbf{\tilde{H}\tilde{H}}^{\rm H}\right),
\end{equation}
where \({\rm SNR}\triangleq {P_{\rm T}{\rm P_L}}/{\sigma^2}\). It is well known that this metric is maximized when \(\mathbf{\tilde{H}\tilde{H}}^{\rm H}\) is diagonal. In conventional MIMO systems, this diagonalization is performed via SVD precoding/combining and power allocation based on the waterfilling principle. For this paper's context, it will be next shown that, without precoding, similar diagonalization and power allocation can be achieved OTA through appropriate RIS configuration. To this end, \(\mathcal{C}(\boldsymbol{\phi})\) can be rewritten using Sylvester’s determinant identity after substituting \(\mathbf{{H}}^{\rm H}_{\rm R}\) and \(\mathbf{H}_{\rm T}\) with their SVDs, as follows:
\begin{align}
    \mathcal{C}(\boldsymbol{\phi})=\log_2 {\rm det}\biggl(& \mathbf{I}_{N_{\rm R}} + \frac{{\rm SNR}}{N_{\rm T}}  \mathbf{D}_{\rm R}  \mathbf{V}_{\rm R}^{\rm H} \boldsymbol{\Phi} \mathbf{U}_{\rm T}  \mathbf{D}_{\rm T}   \nonumber\\
    &\times\left(\mathbf{D}_{\rm R}  \mathbf{V}_{\rm R}^{\rm H} \boldsymbol{\Phi} \mathbf{U}_{\rm T}  \mathbf{D}_{\rm T} \right)^{\rm H}\biggl) \label{eq: capacity in SVD form}.
\end{align}
As per Theorem~\ref{theorem: channel magnitude convergence}, when the RIS phase vector \(\boldsymbol{\phi}\) is configured to maximize a specific entry of \(\mathbf{V}_{\rm R}^{\rm H}\boldsymbol{\Phi}\mathbf{U}_{\rm T}\), all remaining entries asymptotically converge to zero. Hence, we can use the RIS to \textit{asymptotically diagonalize the channel} by partitioning its response-tunable elements to maximize the diagonal entries of the effective channel matrix \(\mathbf{H}_{\rm eff}\triangleq\mathbf{D}_{\rm R}  \mathbf{V}_{\rm R}^{\rm H} \boldsymbol{\Phi} \mathbf{U}_{\rm T}  \mathbf{D}_{\rm T}\). Consequently, letting \(N_{\min}\triangleq\min\{N_{\rm T},N_{\rm R}\}\) and \(N_{\max}\triangleq\max\{N_{\rm T},N_{\rm R}\}\), we present in the following theorem the conditions for the RIS to perform OTA diagonalization, i.e., for \(\mathbf{H}_{\rm eff}\) to converge to a diagonal matrix. It is noted that, tn the case of non-square matrices, the term \textit{diagonal} implies that only the first $N_{\min}$ main diagonal entries are nonzero.

\begin{theorem}\label{theorem: channel diagonalization}
    Consider \(\mathbf{H}^{\rm H}_{\rm R}\) and \(\mathbf{H}_{\rm T}\) to be Ricean distributed and the 1-bit RIS to be optimized via SA to maximize the diagonal terms of \(\mathbf{H}_{\rm eff}\). Then, the singular values of \(\mathbf{H}_{\rm eff}\) (and \(\mathbf{\tilde{H}}\)) converge to the absolute diagonal values of \(\mathbf{H}_{\rm eff}\) in the regime where \(N_{\min},N_{\max},N_{\rm S}\to \infty\) with \(\frac{N_{\min} \sqrt{N_{\max}}}{\sqrt{N_{\rm S}}}\to 0\).
\end{theorem}

\begin{proof}
    We begin by assuming that the RIS elements are uniformly distributed to optimize each of the diagonal entries, i.e., \(N_{\rm S}/N_{\min}\) elements are allocated per diagonal entry. We will prove that, for the \(i\)-th diagonal element, it holds that \(|[\mathbf{H}_{\rm eff}]_{i,i}|\gg \max[\|[\mathbf{H}_{\rm eff}]_{i+1:N_{\rm R},i+1}\|_{\rm F}, \|[\mathbf{H}_{\rm eff}]_{i+1,i+1:N_{\rm T}}\|_{\rm F}]\), implying that the magnitude of each column and row originating from the \(i\)-th diagonal element is asymptotically dominated by the diagonal entry itself.

    The magnitude of the first diagonal element \(d_{{\rm R},1}d_{{\rm T},1}\left(\mathbf{v}^*_{{\rm R},1}\odot \mathbf{u}_{{\rm T},1}\right)^{\rm T}\boldsymbol{\phi}\) scales with \(N_{\rm S}\sqrt{N_{\rm T}N_{\rm R}}\) when all RIS elements are used to maximize it. In the case where only \(N_{\rm S}/N_{\min}\) elements are used, this scaling is reduced by a factor of \(1/N_{\min}\) (this is a direct extension of \cite[Theorem~2]{stylianopoulos2024signalignment}). Moreover, from Theorem~\ref{theorem: channel magnitude convergence}, we know that unoptimized terms \(\left(\mathbf{v}^*_{{\rm R},i}\odot \mathbf{u}_{{\rm T},j}\right)^{\rm T}\boldsymbol{\phi}\) scale as \(\mathcal{O}(1/\sqrt{N_{\rm S}})\), hence, accounting for the magnitude of \(d_{{\rm R},i}d_{{\rm T},j}\), we find that unoptimized values in the first column and row scale as \(\mathcal{O}\left(\sqrt{N_{\rm S}N_{\rm T}}\right)\) and \(\mathcal{O}\left(\sqrt{N_{\rm S}N_{\rm R}}\right)\), respectively. Thus, the norms of the first column and row (excluding the first element) scale identically as \(\mathcal{O}\left(\sqrt{N_{\rm R} N_{\rm S}N_{\rm T}}\right)\). For the first element to asymptotically dominate its row and column, it must hold that:
    \( \frac{N_{\rm S}\sqrt{N_{\rm T}N_{\rm R}}}{N_{\min}} \gg \sqrt{N_{\rm R} N_{\rm S}N_{\rm T}}
    \Rightarrow \sqrt{N_{\rm S}} \gg {N_{\min }}\).
    Applying the same reasoning to the second diagonal element yields the condition \(N_{\rm S}/N_{\min}\gg \sqrt{N_{\rm S}N_{\max}} \Rightarrow \sqrt{N_{\rm S}}\gg N_{\min}\sqrt{N_{\max}}\). Since all \(i\geq 2\) diagonal elements scale identically w.r.t. \(N_{\rm S}\), this condition suffices to ensure that all diagonal entries dominate the respective off-diagonal terms.
    Capitalizing on this fact yields:
    \(
    \|\mathbf{\tilde{H}}\|^2_{\rm F} = \|\mathbf{H}_{\rm eff}\|^2_{\rm F} = \sum_{i,j=1}^{N_{\min}} |[\mathbf{H}_{\rm eff}]_{i,j}|^2 \to \sum_{i=1}^{N_{\min}} |[\mathbf{H}_{\rm eff}]_{i,i}|^2,
    \) having the normalized error:
    \(
    \left|\frac{\|\mathbf{\tilde{H}}\|^2_{\rm F} - \sum_{i=1}^{N_{\min}} |[\mathbf{H}_{\rm eff}]_{i,i}|^2}{\|\mathbf{\tilde{H}}\|^2_{\rm F}}\right| \to\mathcal{O}\left(\frac{N_{\min}^2N_{\max}}{N_{\rm S}}\right).
    \)

    Then, using Mirsky’s theorem \cite{mirsky1960inequality} which bounds the difference between the singular values of two matrices by the Frobenius norm of their difference, leads to
    \(
    \|\sigma(\mathbf{H}_{\rm eff}) - \sigma(\boldsymbol{\Lambda})\|_2 \leq \|\mathbf{H}_{\rm eff} - \boldsymbol{\Lambda}\|_{\rm F},
    \)
    where \(\boldsymbol{\Lambda}\) is the (non-square) diagonal matrix containing the diagonal elements of \(\mathbf{H}_{\rm eff}\). Since the off-diagonal terms asymptotically vanish under the stated condition, this implies that \(\mathbf{H}_{\rm eff}\)'s singular values converge to those of the diagonal matrix \(\boldsymbol{\Lambda}\), with the same normalized error as the Frobenius norm approximation presented previously. Regarding non-uniform element allocation, the latter results hold under the mild assumption that the number of RIS elements allocated to each stream differs by a finite amount, or that the diagonal entries that are not maximized are significantly smaller than $N_{\min}$. Finally, in the case where only a finite number $d \ll N_{\min}$ of streams have a nonzero number of allocated elements, each optimized diagonal element asymptotically overpowers the remaining matrix $[\mathbf{H}_{\rm eff}]_{d+1:N_{\rm R},\, d+1:N_{\rm T}}$ if:
    \(\frac{N_{\rm S}}{d} \gg N_{\max} \sqrt{N_{\rm S}} \quad \Rightarrow \quad \sqrt{N_{\rm S}} \gg d N_{\max}\). However, the \(N_{\min}-d\) remaining singular values cannot be inferred from the diagonal entries of \(\mathbf{H}_{\rm eff}\); nevertheless, their magnitude is known to be bounded by \(\mathcal{O}\left(N_{\max}\sqrt{N_{\rm S}}\right)\). 
\end{proof}
\vspace{-0.2 cm}
\begin{remark} Theorems~\ref{theorem: channel magnitude convergence} and~\ref{theorem: channel diagonalization} generalize to the case of RISs with continuous phase responses, for which the same scaling laws apply. To prove this, one can replace SA with phase alignment, i.e., instead of aligning the signs of the~real/imaginary parts of the involved vectors, their phases are fully aligned.
\end{remark}
%but are presented here for the $1$-bit case since it is considered a more reasonable choice for large RISs.

Following this ``diagonalization'' property of large RISs, the capacity can be re-expressed at the regime where Theorem~\ref{theorem: channel diagonalization} holds as a straightforward function of \(\boldsymbol{\phi}\). Specifically, we replace the eigenvalues of \(\mathbf{H}_{\rm eff}\mathbf{H}^{\rm H}_{\rm eff}\) with the squared magnitudes of \(\mathbf{H}_{\rm eff}\)'s diagonal elements and approximate \(\mathcal{C}(\boldsymbol{\phi})\) with:
%\vspace{-0.2 cm}
\begin{align}\label{eq: capacity with diagonalization}
    \hat{\mathcal{C}}(\boldsymbol{\phi})\!\triangleq\!\sum_{i=1}^{N_{\min}}\!\!\log_2\!\!\left(\!1\!+\! \frac{\rm SNR}{N_{\rm T}} d^2_{{\rm R},i}d^2_{{\rm T},i} \!\left| \left( \mathbf{v}_{{\rm R},i}^{*} \odot \mathbf{u}_{{\rm T},i} \right)^{\!\rm T} \!\!\boldsymbol{\phi}\right|^2\right)\!\!.
\end{align}

Furthermore, as discussed in Theorem~\ref{theorem: channel magnitude convergence}, there are two cases for the distributions of \(\mathbf{v}^*_{{\rm R},i}\) and \(\mathbf{u}_{{\rm T},i}~\forall i=1,2,\ldots,N_{\min}\): \textit{i}) deterministic LoS components; and \textit{ii}) complex normal vectors \eqref{eq: Eiganvector Distribution}. For either case, when SA is applied to maximize a specific term, \(|\left( \mathbf{v}_{{\rm R},i}^{*} \odot \mathbf{u}_{{\rm T},i} \right)^{\rm T} \boldsymbol{\phi}|^2 \geq 0.25\) holds~\cite[Theorems~2 and~3]{stylianopoulos2024signalignment}. Accordingly, if only a portion \(\sqrt{p_i}N_{\rm S}\) of the \(N_{\rm S}\) elements is used to optimize the \(i\)-th term, then, it holds that \(| \left( \mathbf{v}_{{\rm R},i}^{*} \odot \mathbf{u}_{{\rm T},i} \right)^{\rm T} \boldsymbol{\phi}|^2\geq 0.25p_i\). Hence, assuming that SA is employed for the RIS design and that the RIS is partitioned into $\sqrt{p_i}N_{\rm S}$ elements per stream, the capacity is lower bounded by
\(
\mathcal{C}(\boldsymbol{\phi}) \geq \sum_{i=1}^{N_{\min}} \log_2 \left(1 + 0.25 \frac{{\rm SNR}}{N_{\rm T}} d_{{\rm R},i}^2 d_{{\rm T},i}^2 p_i \right).
\)
Note that this lower bound is \emph{deterministic} since the singular values have been shown to converge to deterministic limits (see \eqref{eq: convergence of all eigenvalues}); thus, it depends only on the element allocation $p_i$, the channels' Ricean factors, and the sizes of the RIS and transceiver arrays. Exploiting this expression, we henceforth focus on optimizing the RIS element allocation, i.e., the fraction of elements assigned to each stream, by solving:
%\vspace{-0.2cm}
\begin{align*}%\label{eq: Capacity optimization waterfilling-like}
    \mathcal{OP}_2\!\!:\max_{p_i\,\forall i}\!\sum_{i=1}^{N_{\min}}\!\!\log_2\left(\!1\!+\! 0.25\frac{\rm SNR}{N_{\rm T}} d^2_{{\rm R},i}d^2_{{\rm T},i} p_i\!\right)\,{\text{s.t.}}\! \sum_{i=1}^{N_{\min}}\!\!\sqrt{p_i}\!=\!1.
\end{align*}

Interestingly, the latter problem can be solved without requiring instantaneous CSI. Additionally, it can be easily seen that $\mathcal{OP}_2$ resembles the well known waterfilling framework~\cite[Section~4.4.2]{heath}, with the main difference being the constraint, which in the present case involves the square roots of \(p_i\)'s. This structural difference complicates the problem, since, although the objective function is concave, the constraint is also concave. As a result, the feasible set is non-convex, and global optimality cannot be guaranteed using standard convex optimization tools. To address this issue, we employ a Successive Convex Approximation (SCA) strategy, solving iteratively a sequence of convex problems formed by locally linearizing the non-convex constraint.  Specifically, at iteration \( t + 1 \), we approximate the constraint \( \sum_{i=1}^{N_{\min}} \sqrt{p_i^{(t+1)}} = 1 \) via a first-order Taylor expansion around the current points \( p_i^{(t)} \)'s, which correspond to the solutions obtained at iteration \( t \). This yields the following convex surrogate constraint:
%\vspace{-0.2 cm}
\begin{equation}\label{eq: Waterfilling surrogate constraint}
\sum_{i=1}^{N_{\min}} \underbrace{\frac{1}{2\sqrt{p_i^{(t)}}}}_{\triangleq c_i} p_i^{(t+1)} = \underbrace{1 - \sum_{i=1}^{N_{\min}} \left( \sqrt{p_i^{(t)}} - \frac{p_i^{(t)}}{2\sqrt{p_i^{(t)}}} \right)}_{\triangleq\gamma}.
\end{equation}
%where each weight is defined as $c_i \triangleq 0.5 \left(p_i^{(t)}\right)^{-0.5}$.
%\(c_i \triangleq \frac{1}{2\sqrt{p_i^{(t)}}}\) and \(\gamma \triangleq 1 - \sum_{i=1}^{N_{\min}} \left( \sqrt{p_i^{(t)}} - \frac{p_i^{(t)}}{2\sqrt{p_i^{(t)}}} \right)\).
The solution to this generalized waterfilling case is obtained by applying the Karush-Kuhn-Tucker (KKT) conditions to the resulting convex subproblem. This leads to the semi-closed-form update \(p_i^{(t+1)} = \left[ \frac{1}{\eta c_i} - \frac{4 N_{\rm T}}{{\rm SNR}\, d^2_{{\rm R},i} d^2_{{\rm T},i}} \right]^+\!\)
% \begin{equation}\label{eq: Power level update waterfilling}
% p_i^{(t+1)} = \left[ \frac{1}{\eta c_i} - \frac{4 N_{\rm T}}{{\rm SNR}\, d^2_{{\rm T},i} d^2_{{\rm T},i}} \right]^+,
% \end{equation}
with \( \eta \) being the effective water level and \(a^{+}\triangleq \max[0,a]\). The optimal \( \eta \) can be computed by substituting the latter expression into~\eqref{eq: Waterfilling surrogate constraint} and then solving the resulting equation. The iterations stop when the change in the objective function between successive iterations falls below a predefined threshold \(\epsilon\), while ensuring that the original non-convex constraint is met. When finalizing the allocation of the RIS elements, their configuration is designed to follow the SA strategy for each stream. It is noted that the arrangement of elements is irrespective of the resulting capacity in the asymptotic regime. The complexity of the proposed element allocation scheme is \(\mathcal{O}(J_w N_{\min}\log_2(N_{\min}))\), where \(J_w\) are the iterations needed for the SCA convergence.
\begin{remark}
The number of non-zero \(p_i\)'s, denoted as \(N_p\), determines the rank up to which one needs to estimate the left, \(\mathbf{V}_{\rm R}^{\rm H}\), and right, \(\mathbf{U}_{\rm T}\), singular spaces of the \(\mathbf{H}_{\rm R}\) and \(\mathbf{H}_{\rm T}\) channels to configure the RIS, whereas \(N_p\) itself is determined by the Ricean factors and array dimensions. This also indicates that only the channel terms that directly multiply the RIS (i.e., $\mathbf{V}_{\rm R}^{\rm H}$ and $\mathbf{U}_{\rm T}$) are necessary. Moreover,~only~the~signs of the real and imaginary parts of \(\mathbf{v}_{{\rm R},i}\) and \(\mathbf{u}_{{\rm T},i}\) \(\forall i=1,2,\ldots,N_p\) are needed, implying that this channel features' estimation does not require high accuracy, as it is insensitive to the absolute values of the latter vectors. For example, for multipath far-field \(\mathbf{H}_{\rm T}\) and \(\mathbf{H}_{\rm R}\) channels, the RIS configuration scheme requires knowledge of only the \(N_p\) dominant directions from which signals arrive to and depart from it. Finally, it is worth noting that RISs integrated with sensing capabilities have been lately shown to be capable of estimating exactly these left and right spaces of the TX-RIS and RIS-RX channels~\cite{alexandropoulos2025receivingrissenablingchannel,10042447}.
\end{remark}

All in all, the aforedecribed \textit{asymptotic diagonalization} property of large RISs allows us to approximate the eigenvalues of $\mathbf{H}_{\rm eff}\mathbf{H}_{\rm eff}^{\rm H}$ (which are identical to those of $\mathbf{\tilde{H}}\mathbf{\tilde{H}}^{\rm H}$) from $\mathbf{H}_{\rm eff}$'s diagonal entries. This in turn enables the consideration of \eqref{eq: capacity with diagonalization}, whose accuracy is validated in Fig.~\ref{fig: Capacity Approximation}. Therein, the NMSE between \(\mathcal{C}(\boldsymbol{\phi})\) and \(\hat{\mathcal{C}}(\boldsymbol{\phi})\) is shown averaged over \(200\) Monte Carlo (MC) runs, considering random RIS element allocation with uniform distribution per spatial stream. 

% \begin{figure}[t]
% \centering
%     \includegraphics[width=1\linewidth]{figures/capacity_theorem.eps}
%     \caption{The asymptotic diagonalization of Theorem~\ref{theorem: channel diagonalization} is validated by showing the NMSE between the actual \(\mathcal{C}(\boldsymbol{\phi})\) and the approximated \(\hat{\mathcal{C}}(\boldsymbol{\phi})\) capacities, with the latter assuming a diagonal channel. The system parameters are as follows: \(N_{\min}=N_{\rm R}=N_{\rm T}=16\), \(K_{\rm T}=K_{\rm R}=0\) dB, while \(N_{\rm S}=1:10:100 \times N^3_{\min}\), and the element allocation per stream was performed randomly with uniform distribution and weighted to meet the constraint \(\sum_{i=1}^{N_{\min}}\sqrt{p}_i=1\).}
%     \label{fig: Capacity Approximation}
% \end{figure}

\section{Performance Evaluation}\label{sec: Performance Evaluation}

\begin{figure*}[!t]
    \centering
    % First subfigure
    \begin{subfigure}[b]{0.33\textwidth}
        \centering
        \includegraphics[width=\textwidth]{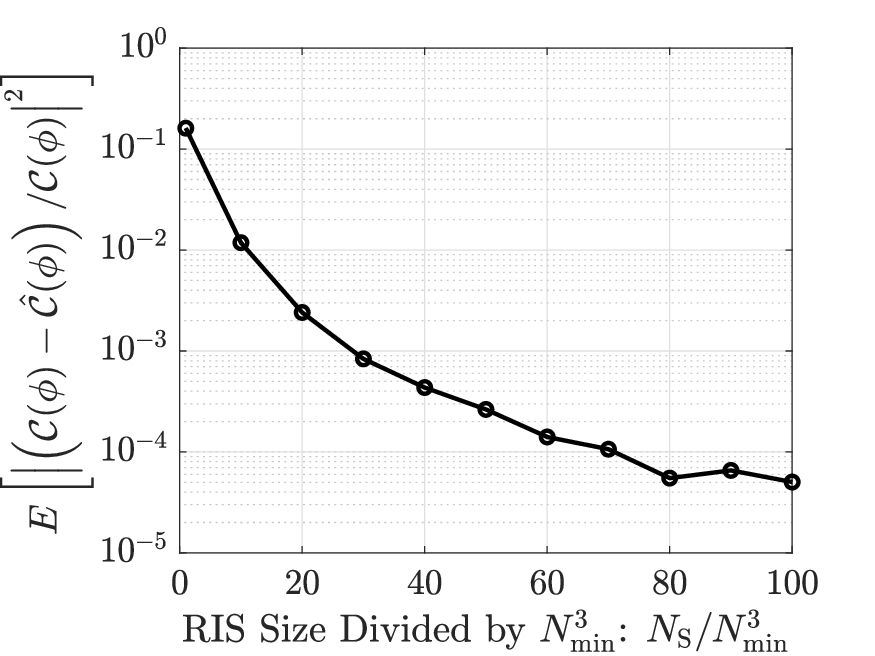} % Replace with your image
        \caption{}
        %\caption{Eigenvalues for a single instance of \(\mathbf{HH}^{\rm H}\).}
        \label{fig: Capacity Approximation}
    \end{subfigure}
    \hfill
    % Second subfigure
    \begin{subfigure}[b]{0.33\textwidth}
        \centering
        \includegraphics[width=\textwidth]{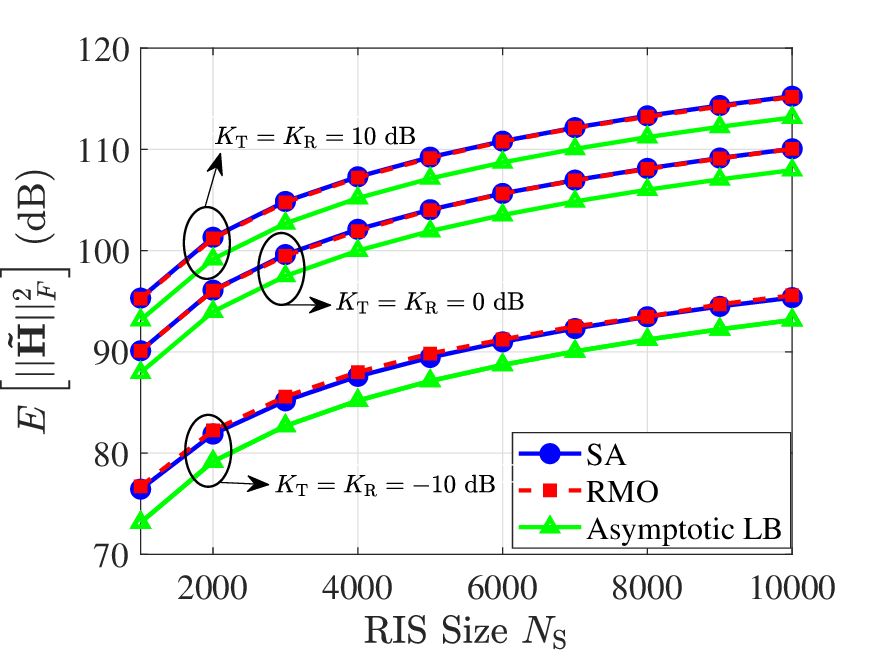} 
        \caption{}
        %\caption{MNSE of eigenvalues for increasing \(N\).}
        \label{fig: Channel Magnitude Optimization}
    \end{subfigure}
    % Third subfigure
    \begin{subfigure}[b]{0.32\textwidth}
        \centering
        \includegraphics[width=\textwidth]{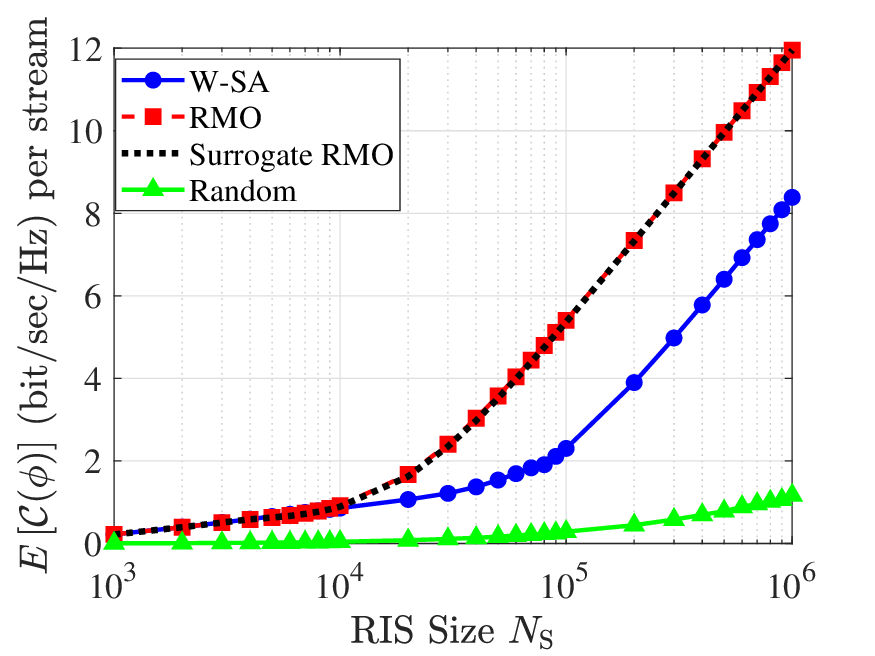} % Replace with your image
        \caption{}
        %\caption{MNSE of principal eigenvalue versus Ricean \(K\).}
        \label{fig: capacity optimisation}
    \end{subfigure}
    \caption{\small 
    The left subfigure (Fig.~\ref{fig: Capacity Approximation}) validates the asymptotic diagonalization property established in Theorem~\ref{theorem: channel diagonalization}, illustrating the NMSE between the actual $\mathcal{C}(\boldsymbol{\phi})$ and the approximated $\hat{\mathcal{C}}(\boldsymbol{\phi})$ capacities, with the latter assuming a diagonal channel; simulation parameters used: $N_{\min}=N_{\rm R}=N_{\rm T}=16$, $K_{\rm T}=K_{\rm R}=0$~dB, and $N_{\rm S}=1\!:\!10\!:\!100\times N_{\max}^3$. Element allocation per stream was random with uniform distribution and weighted to meet $\mathcal{OP}_2$'s constraint. The middle subfigure (Fig.~\ref{fig: Channel Magnitude Optimization}) illustrates the optimized channel magnitude $||\mathbf{\tilde{H}}||_{\rm F}^2$, considering $N_{\rm T}=N_{\rm R}=100$ and $K_{\rm T}=K_{\rm R}=\{-10,0,10\}$~dB. In particular, the performance of SA is compared with RMO and the asymptotic Lower Bound (LB) derived in Section~\ref{sec: Channel Gain Maximization}: $0.25\,\frac{K_{\rm T}K_{\rm R}}{(1+K_{\rm T})(1+K_{\rm R})}N_{\rm S}^2N_{\rm T}N_{\rm R}$. The right subfigure (Fig.~\ref{fig: capacity optimisation}) depicts the achievable rate with the proposed W-SA design against two benchmarks: RMO applied to the original capacity expression~\eqref{eq: capacity formula} and RMO applied to the surrogate formulation~\eqref{eq: capacity with diagonalization}; simulation parameters used: $N_{\rm T}=N_{\rm R}=10$ and $K_{\rm T}=K_{\rm R}=0$~dB. It is shown that, although the RMO variants achieve up to $30\%$ higher capacity than W-SA, the latter requires up to $10^3$ times less computation time (see Table~\ref{tab: running-time capacity}).
    \vspace{-0.2 cm}}
    \label{fig: performance evaluation figures}
\end{figure*}

In this section, we evaluate the proposed approaches in terms of channel gain and capacity performances, benchmarking them against the most established approach for optimizing RISs, namely the Riemannian Manifold Optimization (RMO) framework \cite{Liu2024Optimizationtechniques}. This framework refers to gradient descent/ascent methods that are performed over the complex circle, which is ideal for tuning the phase shifts of RISs. Regarding channel gain maximization, the complexity of performing RMO is \(\mathcal{O}(j_{{\rm RMO}}N_{\rm S} N_{\rm R} N_{\rm T})\), with \(j_{{\rm RMO},1}\) being the iterations till convergence of RMO,
%, whereas \(j_{{\rm RMO},1}\) can scale with \(\mathcal{O}(\sqrt{N}_{\rm S})\), since it depends on the root of the condition number of \((\mathbf{H}^{\rm H}_{\rm R}\mathbf{H}_{\rm R})\odot(\mathbf{H}_{\rm T}\mathbf{H}^{\rm H}_{\rm T})\) \cite{Shewchuk1994RMO_theory}. 
while the complexity of performing SA for the same objective is \(\mathcal{O}(N_{\rm S})\), since it is essentially a closed-form assignment that only requires the element-wise multiplication \( \mathbf{a}^{*}_{\rm S}(\phi^{\rm O}_{\rm S},\theta^{\rm O}_{\rm S})\odot \mathbf{a}_{\rm S}(\phi^{\rm I}_{\rm S},\theta^{\rm I}_{\rm S})\), which is easily parallelizable. For capacity optimization, the asymptotic complexity of W-SA is \(\mathcal{O}(J_w N_{\min}\log_2(N_{\min}) + N^2_{\min}N_{\rm S})\), where the latter term includes the complexity of computing the SVDs of \(\mathbf{H}^{\rm H}_{\rm R}\) and \(\mathbf{H}_{\rm T}\). Performing RMO to maximize the capacity requires \(\mathcal{O}(j_{{\rm RMO}} (N^3_{\min} + N^2_{\min}N_{\rm S}))\), which, although looking comparable, it will be shown that the time complexity, due to it being an iterative scheme, is orders of magnitude larger, compared to the semi-closed-form proposed solution. For both capacity and channel gain optimization with RMO, the RIS configuration was quantized to meet the \(1\)-bit tunability constraint. Lastly, all results presented in this section were averaged over \(200\) MC realizations.

\subsection{Results for Channel Gain Maximization}
According to Theorem~\ref{theorem: channel magnitude convergence}, the optimal strategy for maximizing the channel magnitude is to perform SA on the LoS component of the channel, provided that $K_{\rm T}, K_{\rm R} > 1/N_{\rm T}, 1/N_{\rm R}$, respectively. The corresponding asymptotic Lower Bound (LB) derived for this configuration is $0.25\,\frac{K_{\rm T}K_{\rm R}}{(1+K_{\rm T})(1+K_{\rm R})} N_{\rm S}^2 N_{\rm T} N_{\rm R}$. Figure~\ref{fig: Channel Magnitude Optimization} compares this bound with the simulated performance of SA as the number of RIS elements $N_{\rm S}$ increases, demonstrating that the proposed LB is tight, with only a $2$–$3$~dB deviation. Furthermore, SA is benchmarked against RMO, which is an iterative optimization algorithm, in contrast to the closed-form nature of SA. The results show that performing SA on the LoS component yields near-optimal performance even under non-LoS-dominated conditions:  the curves of SA and RMO align completely and are closely followed by the asymptotic LB. It is also worth noting that RMO entails substantially higher computational complexity %, $\mathcal{O}(j_{{\rm RMO},1} N_{\rm S} N_{\rm T} N_{\rm R})$, 
and requires full CSI, whereas SA relies only on the estimation of the LoS component. The average runtimes are summarized in Table~\ref{tab: running-time channel gain}, indicating that RMO requires approximately $10^5$ times more computation time to achieve identical results with the proposed SA approach.

% \begin{table}[t]
%     \centering
%     \begin{tabular}{c|c|c}
%         \hline
%         \(N_{\rm S}\) & SA (sec) & RMO (sec) \\
%         \hline
%         2000 & \(7.1\times10^{-5}\) & 4.0 \\
%         4000 & \(1.2\times10^{-4}\) & 6.7 \\
%         6000 & \(1.4\times10^{-4}\) & 9.6 \\
%         8000 & \(1.9\times10^{-4}\) & 11.5 \\
%         10000 & \(2.4\times10^{-4}\) & 15.1 \\
%         \hline
%     \end{tabular}
%     \caption{Average running time of performing SA and RMO for channel gain maximization versus \(N_{\rm S}\).}
%     \label{tab: running-time}
% \end{table}

\begin{table}[t]
  \centering
  \setlength{\tabcolsep}{3pt}
  \renewcommand{\arraystretch}{1.05}

  % --- Top subtable: Channel-gain ---
  \begin{subtable}[t]{\linewidth}
    \centering
    \small
    \begin{tabular}{c|c|c}
      \hline
      \(N_{\rm S}\) & SA (sec) & RMO (sec) \\
      \hline
      \(2\times 10^{3}\) & \(7.1\times10^{-5}\) & 4.0 \\
      \(4\times 10^{3}\) & \(1.2\times10^{-4}\) & 6.7 \\
      \(6\times 10^{3}\) & \(1.4\times10^{-4}\) & 9.6 \\
      \(8\times 10^{3}\) & \(1.9\times10^{-4}\) & 11.5 \\
      \(10^{4}\)         & \(2.4\times10^{-4}\) & 15.1 \\
      \hline
    \end{tabular}
    \caption{Channel gain maximization.} \label{tab: running-time channel gain}
  \end{subtable}

  % --- Bottom subtable: Capacity ---
  \begin{subtable}[t]{\linewidth}
    \centering
    \small
    \begin{tabular}{c|c|c|c}
      \hline
      \(N_{\rm S}\) & W-SA (sec) & Surrogate RMO (sec) & RMO (sec) \\
      \hline
      \(5\times 10^{3}\) & \(4.16\times10^{-4}\) & \(3.17\times10^{-2}\) & \(1.62\times10^{-1}\) \\
      \(5\times 10^{4}\) & \(6.00\times10^{-3}\) & \(3.51\times10^{-1}\) & \(2.32\) \\
      \(5\times 10^{5}\) & \(4.85\times10^{-2}\) & \(2.46\)               & \(6.16\) \\
      \(10^{6}\)         & \(9.29\times10^{-2}\) & \(4.75\)               & \(11.5\) \\
      \hline
    \end{tabular}
    \caption{Capacity maximization.}  \label{tab: running-time capacity}
  \end{subtable}

  \caption{Average runtime versus \(N_{\rm S}\) in a computer with a $10$-th generation Intel Core i5 CPU and $32$ GB RAM.}
  \label{tab:runtime-both}
\end{table}

% \begin{figure}[t]
%     \includegraphics[width=1\linewidth]{figures/channel_mag_fig.drawio.pdf}
%     \caption{The optimized magnitude of the channel \(||\mathbf{\tilde{H}}||_{\rm F}^2 \,{\rm (dB)} \) averaged over MC realizations. The performance obtained by performing SA is compared with RMO, and both of those are compared with the asymptotic LB that was introduced in end of Section~\ref{sec: Channel Gain Maximization}: \(0.25\frac{K_{\rm T}K_{\rm R}}{(1+K_{\rm T})(1+ K_{\rm R})} N^2_{\rm S}N_{\rm T}N_{\rm R}\). The system parameters used are \(N_{\rm T}=100,N_{\rm R}=100,\) and \(K_{\rm T}=K_{\rm R}=-10,0,10\) dB.}
%     \label{fig: Channel Magnitude Optimization}
% \end{figure}

\subsection{Results for Capacity Maximization}

According to Theorem~\ref{theorem: channel diagonalization}, maximizing the diagonal entries of $\mathbf{H}_{\rm eff}$ ensures that they asymptotically converge to $\tilde{\mathbf{H}}$'s singular values. Consequently, the capacity can be expressed as a sum over the squared diagonal entries of $\mathbf{H}_{\rm eff}$, enabling the capacity optimization scheme described in Section~\ref{sec: Capacity Optimization}, referred to as W-SA. This proposed scheme is closed-form w.r.t. the RIS size, while the SCA-based iterative step concerns only the element allocation per stream. The latter problem is a low-dimensional one, involving $N_{\min}$ variables, that can be computed offline, depending solely on the Ricean factors and the transceivers and RIS dimensions. As a result, the computational complexity of W-SA scales identically to that of a single RMO iteration. Moreover, W-SA requires only partial CSI knowledge, specifically, the real and imaginary signs of the $N_p$ dominant left and right singular vectors $\mathbf{u}_{\rm T}$ and $\mathbf{v}^*_{\rm R}$.

Figure~\ref{fig: capacity optimisation} compares the capacity performance of the proposed W-SA with two benchmarks: RMO applied to the original capacity expression in \eqref{eq: capacity formula} and RMO applied to the surrogate formulation in \eqref{eq: capacity with diagonalization}. The nearly identical performance of the two RMO variants validates the underlying intuition that, maximizing only the diagonal entries of $\mathbf{H}_{\rm eff}$, suffices for capacity maximization, corroborating the derived capacity approximation. Although RMO achieves slightly higher performance compared to W-SA, with up to $30\%$ difference observed for $N_{\rm S} > 10^4$, this gap arises because the proposed design optimizes each eigenvalue group independently rather than jointly maximizing all eigenvalues. Nevertheless, W-SA drastically reduces computational complexity by reducing the optimization to an iterative scheme over $N_{\min}$ variables and a closed-form step in $N_{\rm S}$. The runtime comparison in Table~\ref{tab: running-time capacity} highlights the efficiency gains: W-SA runs $10^2$–$10^3$ times faster than RMO, while RMO on the surrogate capacity is approximately $3\text{-}7\times$ faster than RMO on the original formulation. This is attributed to the fact that $\hat{\mathcal{C}}(\boldsymbol{\phi})$ avoids determinant computations and involves only summations.

\section{Conclusion}

This paper investigated the asymptotic performance of $1$-bit RIS-aided MIMO communication systems under Ricean fading, focusing on the practically relevant regime where the RIS size scales asymptotically larger than the transceiver dimensions. It was shown that channel gain maximization can be achieved through SA using only the LoS components of TX-RIS and RIS-RX channels, and instantaneous performance guarantees for this strategy were established. Furthermore, the necessary conditions and RIS configuration enabling OTA diagonalization of the RIS-parametrized MIMO channel were derived. This property facilitated the formulation of a deterministic LB for the instantaneous capacity motivating a waterfilling-inspired capacity optimization framework based on RIS element allocation per communication stream. Overall, this work provides new insights into the optimization and scaling laws of RIS-aided MIMO systems, demonstrating that large RISs can inherently diagonalize the effective channel, alleviating this burden from the precoder, while offering performance guarantees and low complexity configuration.

\FloatBarrier
\bibliographystyle{IEEEtran}
\vspace{-0.1cm}
\bibliography{references}

\end{document}